\documentclass{article}
\usepackage{float}
\usepackage{amsmath}
\usepackage{amsthm}
\usepackage{amssymb}
\usepackage{graphicx,cite}

\floatstyle{ruled}
\newfloat{algorithm}{tbp}{loa}
\providecommand{\algorithmname}{Algorithm}
\floatname{algorithm}{\protect\algorithmname}

\theoremstyle{definition}
\newtheorem{defn}{\protect\definitionname}
\theoremstyle{plain}
\newtheorem{prop}{\protect\propositionname}
\theoremstyle{plain}
\newtheorem{thm}{\protect\theoremname}
\theoremstyle{plain}
\newtheorem{lem}{\protect\lemmaname}

\usepackage[accepted]{icml2020}

\usepackage{dsfont}

\renewcommand\P{{\mathbb{P}}}

\usepackage{amsmath}
\newcommand\numberthis{\addtocounter{equation}{1}\tag{\theequation}}

\providecommand{\definitionname}{Definition}
\providecommand{\lemmaname}{Lemma}
\providecommand{\propositionname}{Proposition}
\providecommand{\theoremname}{Theorem}

\begin{document}


\twocolumn[
\icmltitle{My Fair Bandit: Distributed Learning of Max-Min Fairness with Multi-player Bandits}

\begin{icmlauthorlist}
\icmlauthor{Ilai Bistritz}{stan}
\icmlauthor{Tavor Z. Baharav}{stan}
\icmlauthor{Amir Leshem}{barIlan}
\icmlauthor{Nicholas Bambos}{stan}
\end{icmlauthorlist}

\icmlaffiliation{stan}{Department of Electrical Engineering, Stanford University}
\icmlaffiliation{barIlan}{Faculty of Engineering, Bar-Ilan University}

\icmlcorrespondingauthor{Ilai Bistritz}{bistritz@stanford.edu}
\icmlcorrespondingauthor{Tavor Z. Baharav}{tavorb@stanford.edu}

\icmlkeywords{Multi-player bandits, fairness, resource allocation}

\vskip 0.3in
]

\printAffiliationsAndNotice{}  

\begin{abstract}
Consider $N$ cooperative but non-communicating players where each plays one out of $M$ arms for $T$
turns. Players have different utilities for each arm, representable
as an $N\times M$ matrix. These utilities are unknown to the
players.
In each turn players select an arm and receive a noisy observation of their utility for it. However, if any other players selected the same arm that turn, all colliding players will receive zero utility due to the conflict.
No other communication or coordination between the players
is possible. Our goal is to design a distributed algorithm that learns
the matching between players and arms that achieves max-min fairness
while minimizing the regret. We present an algorithm and prove that
it is regret optimal up to a $\log\log T$ factor. This is the first max-min fairness multi-player bandit algorithm with (near) order optimal regret. 
\end{abstract}

\section{Introduction}

In online learning problems, an agent sequentially makes
decisions and receives an associated reward. When this reward is
stochastic, the problem takes the form of a stochastic multi-armed
bandit problem \cite{bubeck2012regret}. However, stochastic bandits assume
stationary reward processes that are rarely the case in practice; they are too optimistic about the environment.
To deal with this shortcoming, one can model the rewards as being determined by
an adversary, which leads to a formulation known as non-stochastic or adversarial bandits
\cite{bubeck2012regret}. Naturally, the performance guarantees against
such a powerful adversary are much weaker than in the stochastic case, as
adversarial bandits are usually overly pessimistic about the environment. 
Is there an alternative that lies in the gap between the two? 

Multi-player bandits is a promising answer, which has seen a surge of interest in recent years
\cite{Liu2010,Vakili2013,Lai2008,Anandkumar2011,liu2019competing,magesh2019multi,Liu2013,Avner2014,sankararaman2019social,Evirgen2017,Cohen2017,Avner2016}.
One primary motivation for studying multi-player bandits is distributed resource allocation.
Examples include channels in communication networks, computation resources on servers, consumers and items, etc. 
In most applications of interest, the reward of a player is a stochastic function of the decisions of other players that operate in the same environment.
Thinking of arms as resources, we see that while these players are not adversaries, conflicts still arise due to the players' preferences among the limited resources. To model that, we assign zero reward to players that choose the same arm. 
The goal of multi-player bandit algorithms is to provide a distributed
way to learn how to share these resources optimally in an online manner. This is
useful in applications where agents (players) follow a standard or protocol,
like in wireless networks, autonomous vehicles, or with a team of robots. 

A common network performance objective is the sum of rewards of the
players over time. As such, maximizing the sum of rewards has received the vast majority of the attention in the multi-player bandit literature
\cite{hanawal2018multi,Besson2018,tibrewal2019distributed,bistritz2018distributed,Kalathil2014,Nayyar2016,boursier2019sic,boursier2019practical}.
However, in the broader literature of network optimization, the sum
of rewards is only one possible objective. 
One severe drawback of this objective is that it has no fairness guarantees.
As such, the maximal sum of rewards
assignment might starve some users. In many applications,
the designer wants to make sure that all users will enjoy at least
minimal target Quality of Service (QoS).

Ensuring fairness has recently been recognized by the machine learning community as a problem of
key importance. In addition to interest in fair classifiers
\cite{zemel2013learning}, fairness has been recognized as a major design parameter
in reinforcement learning and single-player bandits as well \cite{jabbari2017fairness,wei2015mixed,joseph2016fairness}. Our work addresses this major concern for the emerging field of multi-player bandits.

In the context of maximizing the sum of rewards, many works on multi-player
bandits have considered a model where all players have the same vector
of expected rewards \cite{Rosenski2016,boursier2019sic,alatur2019multi,bubeck2019non}. While being relevant in
some applications this model is not rich enough to study fairness, as then the worst off player is simply the one that was allocated the worst resource.
To study fairness a heterogeneous model is necessary, where players have different expected rewards
for the arms (a matrix of expected rewards).
In this case, a fair allocation may prevent
some players from getting their best arm in order to significantly
improve the allocation for less fortunate players. 

Despite being a widely-applied objective in the broader resource allocation literature  \cite{radunovic2007unified,zehavi2013weighted,asadpour2010approximation}, max-min fairness in multi-player bandits has yet to be studied. Some bandit works have studied alternative objectives
that can potentially exhibit some level of fairness  \cite{darak2019multi,bar2019individual}.
In the networking literature, a celebrated notion of fairness is \textit{$\alpha$-fairness} \cite{mo2000fair_alpha} where $\alpha=1$ yields proportional fairness and $\alpha=2$ yields sum of utilities. While for constant $\alpha$, $\alpha$-fairness can be maximized in a similar manner to \cite{bistritz2018distributed}, the case of max-min fairness corresponds to $\alpha\rightarrow \infty$ and is fundamentally different.

Learning to play the max-min fairness allocation involves major technical
challenges that do not arise in the case of maximizing the sum of
rewards (or in the case of $\alpha$- fairness). The sum-rewards
optimal allocation is unique for ``almost all'' scenarios (randomizing the expected rewards).
This is not the case with max-min fairness, as there will typically be multiple optimal allocations.
This complicates
the distributed learning process, since players will have to agree on a specific
optimal allocation to play, which is difficult to do without communication.
Specifically, this rules out using similar techniques to those used
in \cite{bistritz2018distributed} to solve the sum of rewards case
under a similar multi-player bandit setting (i.e., matrix of expected
rewards, no communication, and a collision model). 

Trivially, any matching where all players achieve reward greater than $\gamma$
has max-min value of at least $\gamma$. We observe that finding these matchings, called $\gamma$-matchings in this paper,
can be done via simple dynamics that introduce an absorbing
Markov chain with the $\gamma$-matchings as the absorbing states.
This insight allows for a more robust algorithm than that of \cite{bistritz2018distributed}
which relies on ergodic Markov chains that have an exploration
parameter $\varepsilon$ that has to be tuned. 

In this work we provide an algorithm
that has provably order optimal regret up to a $\log\log T$ factor
(that can be arbitrarily improved to any factor that increases with the horizon
$T$). We adopt the challenging model with heterogeneous
arms (a matrix of expected rewards) and no communication between the
players. The only information players receive regarding their peers is
through the collisions that occur when two or more players pick the
same arm. 

\subsection{Outline}

In Section 2 we formulate
our multi-player bandit problem of learning the max-min matching under the collision model with no communication
between players. In Section 3 we present our distributed max-min fairness algorithm and state our regret bound (proved in Section 8).
Section 4 analyzes the exploration phase of our algorithm. Section
5 analyzes the matching phase of our algorithm and bounds the probability of the exploitation error event. Section 6 presents simulation
results that corroborate our theoretical findings and demonstrate that
our algorithm learns the max-min optimal matching faster than our analytical
bounds suggest. Finally, Section 7 concludes the paper. 

\section{Problem Formulation}
We consider a stochastic game played by a set of $N$ players $\mathcal{N}=\left\{ 1,...,N\right\} $
over a finite time horizon $T$. The strategy
space of each player is the set of $M$ arms with indices denoted
by $i,j\in\{1,...,M\}$. We assume that $M\geq N$, since otherwise the
max-min utility is trivially zero. The horizon $T$ is not known by any of the players, and is considered to be much larger than $M$ and $N$ since we assume that the game is played for a long time. 
Let $t$ be the discrete turn index.  At each turn $t$, all players
simultaneously pick one arm each. The arm that player $n$ chooses
at turn $t$ is $a_{n}\left(t\right)$ and the strategy profile (vector of arms selected) at
turn $t$ is $\boldsymbol{a}\left(t\right)$. Players do not know
which arms the other players chose, and need not even know the number of players $N$.

Define the no-collision indicator of arm $i$ in strategy profile
$\boldsymbol{a}$ to be
\begin{equation}
\eta_{i}\left(\boldsymbol{a}\right)=
\begin{cases}
0 & \Bigl|\mathcal{N}_{i}\left(\boldsymbol{a}\right)\Bigr|>1\\
1 & otherwise.
\end{cases}
\label{eq:1}
\end{equation}
where $\mathcal{N}_{i}\left(\boldsymbol{a}\right)=\left\{ n\,|\,a_{n}=i\right\} $
is the set of players that chose arm $i$ in strategy profile $\boldsymbol{a}$. 
The instantaneous utility of player $n$ at time $t$ with strategy profile $\boldsymbol{a}\left(t\right)$ is
\begin{equation}
\upsilon_{n}\left(\boldsymbol{a}\left(t\right)\right)=r_{n,a_{n}\left(t\right)}\left(t\right)\eta_{a_{n}\left(t\right)}\left(\boldsymbol{a}\left(t\right)\right)\label{eq:2}
\end{equation}
where $r_{n,a_{n}\left(t\right)}\left(t\right)$ is a random reward
which is assumed to have a continuous distribution on $\left[0,1\right]$.
The sequence of rewards of arm $i$ for player $n$, $\left\{ r_{n,i}\left(t\right)\right\} _{t=1}^T$,
 is i.i.d. with expectation
$\mu_{n,i}$. 

An immediate motivation for the collision model above is channel allocation in wireless networks, where the transmission of one user creates interference for other users on the same channel and causes their transmission to fail. Since coordinating a large number of devices in a centralized manner is infeasible, distributed channel allocation algorithms are desirable in practice. In this context, our distributed algorithm learns over time how to assign the channels (arms) such that the maximal QoS guarantee is maintained for all users. However, the collision model is relevant to many other resource allocation scenarios where the resources are discrete items that cannot be shared. 

Next we define the total expected regret for this problem. This is the expected regret that the cooperating but non-communicating players accumulate over time from not playing the optimal max-min allocation. 
\begin{defn}
The total expected regret is defined as
\begin{equation}
R\left(T\right)=\sum_{t=1}^{T} \left(\gamma^{*} - \underset{n}{\min}\ \mathbb{E}\left\{ \upsilon_{n}\left(\boldsymbol{a}\left(t\right)\right)\right\}\right) 
\label{eq:3}
\end{equation}
where $\gamma^{*}=\underset{\boldsymbol{a}}{\max}\ \underset{n}{\min }\ \mathbb{E}\left\{ \upsilon_{n}\left(\boldsymbol{a}\right)\right\}$.
The expectation is over the randomness of the rewards $\left\{ r_{n,i}\left(t\right)\right\} _{t}$, that dictate the random actions $\left\{ a_{n}\left(t\right)\right\}_{t}$.

\end{defn}

Note that replacing the minimum in \eqref{eq:3} with a sum over the players yields the regret for the sum-reward objective case, after redefining $\gamma^{*}$ to be the optimal sum-reward \cite{bistritz2018distributed,Kalathil2014,Nayyar2016,tibrewal2019distributed,boursier2019sic}.

Rewards with a continuous distribution are 
natural in many applications (e.g., SNR in wireless networks). However,
this assumption is only used to argue that since the probability for
zero reward in a non-collision is zero, players can properly estimate their expected rewards. In the case
where the probability of receiving zero reward is not zero, we can assume
instead that each player can observe their no-collision indicator in addition
to their reward. This alternative assumption requires no modifications to our algorithm or analysis. 
Observing one bit of feedback signifying whether other players chose
the same arm is significantly less demanding than observing the actions of other players. 

According to the seminal work in \cite{Lai1985}, the optimal regret
of the single-player case is $O\left(\log T\right)$.
The next proposition shows that $\Omega\left(\log T\right)$ is a lower bound for our multi-player bandit case, since any multi-player bandit algorithm can be used as a single-player algorithm by simulating other players. 
\begin{prop}
\label{Prop Lower Bound}The total expected regret as defined in \eqref{eq:3} of any algorithm is at least $\Omega\left(\log T\right)$.
\end{prop}
\begin{proof}
For $N=1$, the result directly follows from \cite{Lai1985}. Assume that for $N>1$ there is a policy that results in total expected regret better than
$\varOmega\left(\log T\right)$. Then any single player, denoted player $n$,
can simulate $N-1$ other players such that all their expected rewards are larger than her maximal expected reward.
Player $n$ can also generate
the other players' random rewards, that are independent of the actual rewards she
receives. Player $n$ also simulates the policies for other players,
and even knows when a collision occurred for herself and can assign
zero reward in that case.
In this scenario, the expected reward of player $n$ is the minimal expected reward among the non-colliding players.
This implies that $\gamma^{*}$ is the largest expected reward of player $n$. 
Hence, in every turn $t$ without a collision, the $t$-th term in \eqref{eq:3} is equal to the $t$-th term of the single-player regret of player $n$.
If there is a collision in turn $t$, then the $t$-th term in \eqref{eq:3} is $\gamma^{*}$, which bounds from above the $t$-th term of the single-player regret of player $n$.
Thus, the total expected regret upper bounds the single-player regret of player $n$.
Hence, simulating $N-1$ fictitious players
is a valid single player algorithm that violates the
$\varOmega\left(\log T\right)$ bound, which is a contradiction. We
conclude that the $\varOmega\left(\log T\right)$ bound is also valid for $N>1$.
\end{proof}

\section{My Fair Bandit Algorithm}

In this section we describe our distributed multi-player bandit algorithm
that achieves near order optimal regret for the max-min fairness problem.
The key idea behind our algorithm is a global search parameter $\gamma$
that all players track together (with no communication required).
We define a $\gamma$-matching, which is a matching of players to
arms such that the expected reward of each player is at least $\gamma$.
\begin{defn}
An allocation of arms $\boldsymbol{a}$ is a $\gamma$-matching if and
only if $\underset{n}{\min}\ \mathbb{E}\left\{ \upsilon_{n}\left(\boldsymbol{a}\right)\right\}\geq\gamma$. 
\end{defn}

Essentially, the players want to find the maximal $\gamma$ for which
there still exists a $\gamma$-matching. However, even for a given achievable
$\gamma$, distributedly converging to a $\gamma$-matching is a
challenge. Players do not know their expected rewards, and their coordination is extremely limited. To address
these issues we divide the unknown horizon of $T$ turns into epochs, one starting
immediately after the other. Each epoch is further divided into four
phases. In the $k$-th epoch we have:
\begin{enumerate}
\item \textbf{Exploration Phase} - this phase has a length of $\lceil c_{1}\log (k+1)\rceil$
turns for some $c_{1}\ge4$. It is used for estimating the expectation of the arms. As shown in Section \ref{sec:Exploration}, the exploration phase contributes $O\left(\log\log T\log T\right)$
to the total expected regret. 
\item \textbf{Matching Phase} - this phase has a length of $\lceil c_{2}\log (k+1) \rceil$
turns for some $c_{2}\ge1$. In this phase, players attempt to converge
to a $\gamma_{k}$-matching, where each player plays an arm that
is at least as good as $\gamma_{k}$, up to the confidence intervals
of the exploration phase.
To find the matching,
the players follow distributed dynamics that induce an absorbing Markov
chain with the strategy profiles as states. The absorbing states of this
chain are the desired matchings. When the matching phase is long enough,
the probability that a matching exists but is not found is small.
If a matching does not exist, the matching phase naturally does not converge. As shown in Section \ref{sec:Matching}, the matching phase adds $O\left(\log\log T \log T\right)$
to the total expected regret.  
\item \textbf{Consensus Phase} - this phase has a length of $M$ turns. The
goal of this phase is to let all players know whether the matching
phase ended with a matching, using the collisions for signaling. During this phase, every player who did not end the matching phase with a collision plays their matched arm, while the players that ended the matching phase
with a collision sequentially play all the arms for the next $M$ turns. If players deduce that they collectively converged
to a matching, they note this for future reference in the matching
indicator $S_k$. If the matching phase succeeded, the search parameter is updated
as $\gamma_{k+1}=\gamma_{k}+\varepsilon_{k}$. The step size $\varepsilon_{k}$
 is decreasing such that if even a slightly better
matching exists, it will eventually be found. However, it might be
that no $\gamma_{k}$-matching exists. Hence once in a while, with
decreasing frequency, the players reset $\gamma_{k+1}=0$ in order to
allow themselves to keep finding new matchings. This phase adds $O(M\log T)$ to the total expected
regret.
\item \textbf{Exploitation Phase} - this phase has a length of $\left\lceil c_{3}\left(\frac{4}{3}\right)^{k} \right\rceil$
turns for some $c_{3}\geq1$. During this phase, players play the best recently found matching $\boldsymbol{\tilde{\boldsymbol{a}}}_{k^{*}}$, where $k^{*}$ is the epoch within the last $\frac{k}{2}$ epochs with the largest $\gamma$ that resulted in a matching.
This phase adds a vanishing term (with $T$) to the total expected
regret since players eventually play an optimal matching with exponentially small error probability.
\end{enumerate}
\vspace{-.3cm}
The Fair Bandit Algorithm is detailed in Algorithm \ref{alg:FairBandits}.
Our main result is given next, and is proved in Section 8.
\begin{thm}[Main Theorem]
\label{Main Theorem}Assume that the rewards $\left\{ r_{n,i}\left(t\right)\right\} _{t}$
are independent in $n$ and i.i.d. with $t$, with continuous distributions
on $\left[0,1\right]$ with expectations $\left\{ \mu_{n,i}\right\} $.
Let $T$ be the finite deterministic horizon of the game, which is unknown to the players. Let
each player play according to Algorithm \ref{alg:FairBandits} with
any constants $c_{1},c_{2},c_{3}\geq4$. Then the total expected regret satisfies
\begin{align*}
R\left(T\right)
&\leq C_{0} +\left(M+2\left(c_{1}+c_{2}\right)\log\log_{\frac{4}{3}}\frac{T}{c_{3}}\right)\log_{\frac{4}{3}}\frac{T}{c_{3}}\\
&=O\left(\left(M+\log\log T\right)\log T\right) \numberthis \label{eq:4}
\end{align*}
where $C_{0}$ is a constant independent of $T$ and the $\log\log T$ can be improved to any increasing function of $T$ by changing the lengths of the exploration and matching phases. 
\end{thm}
\vspace{-.1cm}
The purpose of the epoch structure is to address the main challenge arising from having multiple players: coordinating between players without communication. To this end, the players in our algorithm try together to find  a $\gamma$-matching, where $\gamma$ is a mutual parameter that they can all update simultaneously but independently, obviating the need for a central entity. Then the main measure of multi-player ``problem hardness'' is the absorption time $\Bar{\tau}$ of the matching Markov chain (see Lemma \ref{lem:Convergence Lemma}), which is unknown. To find a matching we then need a matching phase with increasing length (to eventually surpass $\Bar{\tau}$), taken to be of length $\lceil c_2\log (k+1) \rceil$ for the $k$-th phase, which alone contributes $O(\log\log T \log T)$ to the total expected regret.  Hence, the coordination challenge dominates the regret. 

The additive constant $C_0$ (see \eqref{eq:23}) is essentially the total regret accumulated during the initial epochs when the confidence intervals were still not small enough compared to the gap $\Delta=\underset{n}{\min}\ \underset{i\neq j}{\min}\left|\mu_{n,i}-\mu_{n,j}\right|$ (formalized in \eqref{eq:12-1}) or when the length of the matching phase was not long enough compared to $\Bar{\tau}$ (formalized in \eqref{eq:15}). Hence, $C_0$ depends on  $\Delta$  and $\Bar{\tau}$. In a simplified scenario when $\Delta$  and $\Bar{\tau}$ are known or can be bounded, the lengths of the exploration and matching phase can be made constant  ($c_1$ and $c_2$) and the confidence intervals in \eqref{eq:10} can be set to $\frac{\Delta}{4}$.  Note that with constant length phases, the $\log\frac{T_{e}\left(k\right)}{5M}$ in \eqref{eq:11} would be replaced with $\frac{M}{\Delta^2}$ and the  $\log (\frac{k}{2}+1)$ in \eqref{eq:15} replaced with a constant. Then $c_1,c_2$ can be chosen such that $C_0=0$, by making the probability in \eqref{eq:11} vanish faster than $\left( \frac{3}{4} \right)^k$ and satisfying \eqref{eq:15}. This amounts to choosing $c_1=O\left(\frac{M}{\Delta^2}+M^2\right)$ and $c_2=O(\Bar{\tau})$, which makes our regret bound in \eqref{eq:4} become $O\left((\frac{M}{\Delta^{2}}+M^2+\Bar{\tau})\log T\right)$, since the $\log\log T $ becomes 1 with constant length phases. Nevertheless, this issue is mainly theoretical since in practice, it is easy to choose large enough $c_1,c_2$ such that $C_0$ is very small across various experiments, as can be seen in our simulations in Section \ref{sec:Simulations}.
\vspace{-.2cm}
\section{Exploration Phase \label{sec:Exploration}}
\vspace{-.1cm}
Over time, players receive stochastic rewards from different arms and
average them to estimate their expected reward for each arm. In each
epoch, only $\lceil c_{1}\log (k+1) \rceil$ turns are dedicated to exploration. However,
the estimation of the expected rewards uses all the previous exploration phases, so the
number of samples for estimation at epoch $k$ is $\Theta(k \log k) $. Since players
only have estimates of the expected rewards, they can never be
sure if a matching is a $\gamma$-matching. The purpose of the exploration
phase is to help the players become more confident over time that the matchings
they converge to in the matching phase are indeed $\gamma$-matchings.

\begin{algorithm}[!t]
\caption{\label{alg:FairBandits}My Fair Bandit Algorithm}
\textbf{Initialization}: Set $V_{n,i}=0$ and $s_{n,i}=0$ for all
$i$. Set reset counter $w=0$ with expiration
$e_{w}=1$. Let $\varepsilon_{0}= 1$.

\textbf{For each epoch $k=1,2,\hdots$}
\begin{enumerate}
\item \textbf{Exploration Phase:}
\begin{enumerate}
\item For the next $\lceil c_{1}\log (k+1)\rceil$ turns:
\begin{enumerate}
\item Play an arm $i$ uniformly at random from all $M$ arms. 
\item Receive $r_{n,i}\left(t\right)$ and set $\eta_{i}\left(\boldsymbol{a}\left(t\right)\right)=0$
if $r_{n,i}\left(t\right)=0$ and $\eta_{i}\left(\boldsymbol{a}\left(t\right)\right)=1$
otherwise. 
\item If $\eta_{i}\left(\boldsymbol{a}\left(t\right)\right)=1$ then update
$V_{n,i}=V_{n,i}+1$ and $s_{n,i}=s_{n,i}+r_{n,i}\left(t\right)$.
\end{enumerate}
\item Estimate the expectation of arm $i$ as $\mu_{n,i}^{k}=\frac{s_{n,i}}{V_{n,i}}$
for each $i=1,...,M$.
\item Construct confidence intervals for each $\mu_{n,i}^{k}$ as $C_{n,i}^{k}=\sqrt{\frac{M}{\log V_{n,i}}}$. 
\end{enumerate}
\item \textbf{Matching Phase: }
\begin{enumerate}
\item Update $w\leftarrow w+1.$ If $w=e_{w}$ then set $\gamma_{k}=0$,
$w=0$, $e_{w}=\left\lceil \frac{k}{3}\right\rceil $ and update
$\varepsilon_k=\frac{1}{1+\log k}$. If $w<e_{w}$ then $\varepsilon_k=\varepsilon_{k-1}$.
\item Let $\mathcal{E}_{n}^{k}=\left\{ i\,|\,\mu_{n,i}^{k}\geq\gamma_{k}-C_{n,i}^{k}\right\} $. 
\item Pick $a_{n}\left(t\right)$ uniformly at random from $\mathcal{E}_{n}^{k}$.
\item For the next $\lceil c_{2}\log (k+1) \rceil$ turns: 
\begin{enumerate}
\item If $\eta_{a_{n}\left(t\right)}\left(\boldsymbol{a}\left(t\right)\right)=1$ then
keep playing the same arm, that is $a_{n}\left(t+1\right)=a_{n}\left(t\right)$.
\item If $\eta_{a_{n}\left(t\right)}\left(\boldsymbol{a}\left(t\right)\right)=0$
then pick $a_{n}\left(t+1\right)$ uniformly at random from $\mathcal{E}_{n}^{k}$.
\end{enumerate}
\item Set $\widetilde{a}_{k,n}=a_{n}\left(t\right)$. 
\end{enumerate}
\item \textbf{Consensus Phase:}
\begin{enumerate}
\item If $\eta_{\widetilde{a}_{k,n}}\left(\widetilde{\boldsymbol{a}}_{k}\right)=1$ then play $\widetilde{a}_{k,n}$ for $M$ turns. 
\item If $\eta_{\widetilde{a}_{k,n}}\left(\widetilde{\boldsymbol{a}}_{k}\right)=0$ then play $a_{n}=1,...,M$ sequentially. 
\item \textit{Matching was found:} If you did not experience a collision in
the last $M$ turns, set $\gamma_{k+1}=\gamma_{k}+\varepsilon_k$ and $S_k=1$, else set $\gamma_{k+1} = \gamma_k$, $S_k=0$.
\end{enumerate}
\item \textbf{Exploitation Phase}: For $\left\lceil c_{3}\left(\frac{4}{3}\right)^{k}\right\rceil$
turns, play $\widetilde{a}_{k^{*},n}$ for the maximal $k^{*}$ such that \vspace{-.1cm}
\[
k^{*}\in\underset{\left\lceil \frac{k}{2}\right\rceil \leq\ell\leq k}{\arg\max}\ \gamma_{\ell}S_\ell
.\]
\end{enumerate}
\textbf{End}
\end{algorithm}
In our exploration phase each player picks an arm uniformly at random.
This type of exploration phase is common in various multi-player bandit
algorithms \cite{Rosenski2016,bistritz2018distributed}. However, the nature of what the players
are trying to estimate is different. With a sum of rewards objective, players just need to improve over time the accuracy of the estimation of the expected rewards. With max-min fairness each player needs to make a hard (binary) decision whether a certain arm has expected reward above or below $\gamma$. After the confidence intervals become small enough, if the estimations do fall within their confidence intervals, players can be confident about this hard decision. Under this success event, where the confidence intervals are small enough,
a matching $\boldsymbol{a}$ is a $\gamma$-matching if all players observe that $\mu_{n,a_n}^{k}\geq\gamma-C_{n,a_n}^{k}$. The next lemma bounds the probability that this success event
does not occur, so the estimation for epoch $k$
failed.
\begin{lem}[Exploration Error Probability]
\label{lem: exploration} Let $\left\{ \mu_{n,i}^{k}\right\} $ be
the estimated reward expectations using all the exploration phases
up to epoch $k$, with confidence intervals $\left\{ C_{n,i}^{k}\right\} $.
Define the minimal gap by
\begin{equation}
\Delta \triangleq \underset{n}{\min}\ \underset{i\neq j}{\min}\left|\mu_{n,i}-\mu_{n,j}\right|.\label{eq:5}
\end{equation}
Define the $k$-th exploration error event as 
\begin{equation}
E_{e,k}=\left\{ \exists n,i\,\bigg|\,\left|\mu_{n,i}^{k}-\mu_{n,i}\right|\geq C_{n,i}^{k}\textnormal{ or }C_{n,i}^{k}\ensuremath{\geq}\frac{\Delta}{4}\right\} .\label{eq:6}
\end{equation}
Then for all $k>k_0$ for a large enough constant $k_0$ we have
\begin{equation}
\P\left(E_{e,k}\right)\leq 3NMe^{-\frac{c_{1}}{6}k}.\label{eq:7}
\end{equation}
\end{lem}
\begin{proof}
After the $k$-th exploration phase, the estimation of the expected
rewards is based on $T_{e}\left(k\right)$ samples, and
\begin{equation}
T_{e}\left(k\right)\geq c_{1}\sum_{i=1}^{k}\log (i+1)\geq c_{1}\frac{k}{2}\log\frac{k}{2}.\label{eq:8}
\end{equation}
Let $A_{n,i}\left(t\right)$ be the indicator that is equal to
one if only player $n$ chose arm $i$ at time $t$.  Define $V_{n,i}$,
as the number of visits of player $n$ to arm $i$ with no collision
up to time $t$, and $V_m=\underset{n,i}{\min}\ V_{n,i}$.
The exploration phase consists of uniform and independent arm
choices, so 
$\P\left(A_{n,i}\left(t\right)=1\right)=\frac{1}{M}\left(1-\frac{1}{M}\right)^{N-1}$.\label{eq:12}
We show that each player pulls each arm many times without collisions. Formally:
\begin{align}
\P\bigg(V_m<\frac{T_{e}\left(k\right)}{5M}\bigg)&= \P\left(\bigcup_{i=1}^{M}\bigcup_{n=1}^{N}\left\{ V_{n,i}<\frac{T_{e}\left(k\right)}{5M}\right\} \right)\nonumber\\
&\underset{\left(a\right)}{\leq} NM\P\left(V_{1,1}<\frac{T_{e}\left(k\right)}{5M}\right)\nonumber\\
&\underset{\left(b\right)}{\leq} NMe^{-2\frac{1}{M^{2}}\left(\left(1-\frac{1}{M}\right)^{N-1}-\frac{1}{5}\right)^{2}T_{e}\left(k\right)}\nonumber\\
&\underset{\left(c\right)}{\leq}NMe^{-\frac{1}{18M^{2}}T_{e}\left(k\right)}
\label{eq:9}
\end{align}
where (a) is a union bound, (b) is Hoeffding's inequality for
Bernoulli random variables and (c) follows since $M\geq N$ and $\left(1-\frac{1}{M}\right)^{M-1}-\frac{1}{5}\geq e^{-1}-\frac{1}{5}>\frac{1}{6}$.
By Hoeffding's inequality for random variables \cite{hoeffding1994probability} on $\left[0,1\right]$
\begin{align*}
&\P\left(\bigcup_{n=1}^{N}\bigcup_{i=1}^{M}\left\{ \left|\mu_{n,i}^{k}-\mu_{n,i}\right|\geq C_{n,i}^{k}\,\right\} \bigg|\,\left\{ V_{n,i}\right\} \right)\\&\hspace{.5cm}\leq\sum_{n=1}^{N}\sum_{i=1}^{M}2e^{-2V_{n,i}\left(C_{n,i}^{k}\right)^{2}}\underset{(a)}{\leq}2NMe^{-2\frac{MV_{m}}{\log V_{m}}}\numberthis \label{eq:10}
\end{align*}
where (a) uses $C_{n,i}^{k}=\sqrt{\frac{M}{\log V_{n,i}}}$ and $V_m\ge 3$. Now note that for all $k>k_0$ for a sufficiently large $k_0$:
\begin{equation}
\sqrt{\frac{M}{{\log\left(\frac{T_{e}(k)}{5M}\right)}}}\leq\sqrt{\frac{M}{{\log\left(\frac{c_{1}\frac{k}{2}\log\frac{k}{2}}{5M}\right)}}}<\frac{\Delta}{4}\label{eq:12-1}
\end{equation}

and therefore $V_m\geq\frac{T_{e}\left(k\right)}{5M}$ implies $\underset{n,i}{\max}\ C_{n,i}^{k}<\frac{\Delta}{4}$. Hence, given \eqref{eq:12-1} the event $\bigcup_{n=1}^{N}\bigcup_{i=1}^{M}\left\{ \left|\mu_{n,i}^{k}-\mu_{n,i}\right|\geq C_{n,i}^{k}\,\right\}$ coincides with $E_{e,k}$, so for all $k>k_0$:
\begin{align*}
    \P&\left(E_{e,k}\bigg|\ V_{m}\geq\frac{T_{e}\left(k\right)}{5M}\right)\\
    &= \P\left(\bigcup_{n=1}^{N}\bigcup_{i=1}^{M}\left\{ \left|\mu_{n,i}^{k}-\mu_{n,i}\right|\geq C_{n,i}^{k}\,\right\} \bigg|\,V_{m}\geq\frac{T_{e}\left(k\right)}{5M}\right)\\
    &\underset{\left(a\right)}{\leq}2NMe^{-2\frac{\frac{T_{e}\left(k\right)}{5}}{\log\frac{T_{e}\left(k\right)}{5M}}}\numberthis\label{eq:12-2}
\end{align*}
where (a) uses the law of total probability with respect to $\{V_{n,i}\}$  with Bayes's rule on $\{V_{m}\geq\frac{T_{e}\left(k\right)}{5M}\}$, using the bound in \eqref{eq:10}.
We conclude that for all $k>k_0$:
\begin{align*}
\P\big(E_{e,k}\big)
&=\P\left(E_{e,k}|\ V_m < \frac{T_{e}\left(k\right)}{5M}\right) \P\left(V_m<\frac{T_{e}\left(k\right)}{5M}\right) \\
&\hspace{.4cm}+\P\left(E_{e,k}|\ V_m\geq\frac{T_{e}\left(k\right)}{5M}\right) \P\left(V_{m}\geq\frac{T_{e}\left(k\right)}{5M}\right)\\
&\leq\P\left(V_m<\frac{T_{e}\left(k\right)}{5M}\right)
+\P\left(E_{e,k}|\ V_m\geq\frac{T_{e}\left(k\right)}{5M}\right)\\
&\underset{(a)}{\leq} NMe^{-\frac{T_{e}\left(k\right)}{18M^{2}}} + 2NMe^{-\frac{2T_{e}\left(k\right)}{5\log\frac{T_{e}\left(k\right)}{5M}}}\numberthis
\label{eq:11}
\end{align*}
where (a) uses \eqref{eq:9} and \eqref{eq:12-2}. Finally, \eqref{eq:7} follows by using  \eqref{eq:8} in \eqref{eq:11} for a sufficiently large $k$.
\end{proof}
\section{Matching Phase\label{sec:Matching}}
In this section we analyze the matching phase, where the goal is to distributedly
find $\gamma$-matchings based on the estimated expected rewards from
the exploration phase. We conclude by upper bounding
the probability that an optimal  $\gamma^{*}$-matching
is not played during the exploitation phase. During the matching phase,
the rewards of the arms are ignored, as only the binary decision of
whether an arm is better or worse than $\gamma$ matters. These binary
decisions induce the following bipartite graph between the $N$ players and $M$ arms:
\begin{defn}
Let $G_k$ be the bipartite graph where edge $\left(n,i\right)$
exists if and only if $\mu_{n,i}^{k}\geq\gamma_{k}-C_{n,i}^{k}$. 
\end{defn}
During the $k$-th matching phase, players follow our dynamics to switch arms in order to
find a $\gamma_{k}$-matching in $G_{k}$. These $\gamma_{k}$-matchings (up to confidence intervals)
are absorbing states in the sense that players stop switching arms
if they are all playing a $\gamma_{k}$-matching. The dynamics of
the players induce the following Markov
chain:
\begin{defn}
\label{Def:MatchingDynamics}Define $\mathcal{E}_{n}^{k}=\left\{ i\,|\,\mu_{n,i}^{k}\geq\gamma_{k}-C_{n,i}^{k}\right\} $.
The transition into $\boldsymbol{a}\left(t+1\right)$
is dictated by the transition of each player $n$: 
\begin{enumerate}
\item If $\eta_{a_{n}\left(t\right)}\left(\boldsymbol{a}\left(t\right)\right)=1$
then $a_{n}\left(t+1\right)=a_{n}\left(t\right)$ with probability
1.
\item If $\eta_{a_{n}\left(t\right)}\left(\boldsymbol{a}\left(t\right)\right)=0$
then $a_{n}\left(t+1\right)=i$ with probability $\frac{1}{\left|\mathcal{E}_{n}^{k}\right|}$
for all $i\in\mathcal{E}_{n}^{k}$.
\end{enumerate}
\end{defn}
Note that the matchings in $G_{k}$ are $\gamma_{k}$-matchings only when the confidence intervals are small enough. 
Next we prove that if a matching exists in $G_{k}$, then
the matching phase will find it with a probability
that goes to one. However, we do not need this probability
to converge to one, but simply to exceed a large enough constant.
\begin{lem}
\label{Lem:MatchingTrial} Let $\mathcal{G}_{N,M}$ be the set of all
bipartite graphs with $N$ left vertices and $M$ right vertices that
have a matching of size $N$. Define the random variable $\tau(G,\boldsymbol{a}(0))$
as the first time the process of Definition \ref{Def:MatchingDynamics}, $\{\boldsymbol{a}(t)\}$,
constitutes a matching of size $N$, starting from $\boldsymbol{a}(0)$. Define

\begin{equation}
\Bar{\tau}=\max_{G\in\mathcal{G}_{N,M}, \boldsymbol{a}(0)}\mathbb{E}\left\{ \tau\left(G,\boldsymbol{a}(0)\right)\right\} .
\end{equation}
If $G_{k}$ permits a matching then the $k$-th matching phase converges
to a matching with probability $p\geq 1-\frac{\Bar{\tau}}{\lceil c_{2}\log (k+1)\rceil}$. 
\end{lem}
\begin{proof}
We start by noting that the process $\boldsymbol{a}\left(t\right)$
that evolves according to the dynamics in Definition \ref{Def:MatchingDynamics}
is a Markov chain. This follows since all transitions are a function
of $\boldsymbol{a}\left(t\right)$ alone, with no dependence on $\boldsymbol{a}\left(t-1\right),...,\boldsymbol{a}\left(0\right)$
given $\boldsymbol{a}\left(t\right)$. Let $\mathcal{\mathcal{M}}$ be a matching in $G_{k}$.
Define $\Phi_{\mathcal{M}}\left(\boldsymbol{a}\right)$ to be the number of players
that are playing in $\boldsymbol{a}$ the arm they are matched to
in $\mathcal{M}$. Observe the process $\Phi_{\mathcal{M}}\left(\boldsymbol{a}\left(t\right)\right)$.
If there are no colliding players, then $\boldsymbol{a}\left(t\right)$
is a matching (potentially different from $\mathcal{M}$) and no player will ever change
their chosen arm. Otherwise, for every collision, at least one of the colliding
players is not playing their arm in $\mathcal{M}$. There is a positive probability
that this player will pick their arm in $\mathcal{M}$ at random and all other players will
stay with the same arm. Hence, if $\boldsymbol{a}\left(t\right)$
is not a matching, then there is a positive probability that $\Phi_{\mathcal{M}}\left(\boldsymbol{a}\left(t+1\right)\right)=\Phi_{\mathcal{M}}\left(\boldsymbol{a}\left(t\right)\right)+1$.
We conclude that every non-matching $\boldsymbol{a}$
has a positive probability path to a matching, making $\boldsymbol{a}\left(t\right)$ an absorbing
Markov chain with the matchings as the absorbing states. By Markov's
inequality
\begin{align*} 
\P\big(\tau&\left(G_{k},\boldsymbol{a}(0)\right)\geq \lceil c_{2}\log (k+1)\rceil\big) \numberthis
\\
& \hspace{.5cm} \leq \frac{\mathbb{E}\left\{ \tau\left(G_{k},\boldsymbol{a}(0)\right)\right\} }{\lceil c_{2}\log (k+1)\rceil} 
\leq\frac{\Bar{\tau}}{\lceil c_{2}\log (k+1)\rceil}. \qedhere
\end{align*}\label{eq:13}
\end{proof}
Intriguingly, the Bernoulli trials stemming from trying to find a matching in $\{G_\ell\}$ over consecutive epochs are \textit{dependent}, as after enough successes, there will no longer be a matching in $G_\ell$, yielding success probability 0. The next Lemma shows that Hoeffding's inequality for binomial random variables still applies as long as there are few enough successes, such that there is still a matching in $G_k$. 

\begin{lem} \label{lem:baharavineq}
Consider a sequence of i.i.d. Bernoulli random variables $X_{1},\hdots,X_L$ with success probability p (or at least $p$ for each trial). 
For $x < Lp$, consider $S_x = \sum_{i=1}^L X_i\mathds{1} \{ \sum_{j<i} X_j < x \}$. Then
\begin{equation}
    \P(S_x < x) \le e^{ -2L\left(p-\frac{x}{L}\right)^2}\label{eq:14}.
\end{equation}
\end{lem}

\begin{proof}

If $S_x=m<x$ then $\sum_{i=1}^L X_i < x$, as otherwise the indicators in $S_x$ of the first $x$ indices $i$ where $X_i=1$ will be active, and so $S_x \ge x$, contradicting $S_x=m<x$. Therefore

\begin{align*}
    \P(S_x<x) &\le \P\left(\sum_{i=1}^L X_i <x\right) \le e^{ -2L\left(p-\frac{x}{L}\right)^2}. \qedhere
\end{align*}
\end{proof}
We conclude this section by proving the main Lemma used to prove Theorem
\ref{Main Theorem}. The idea of the proof is to show that if the
past $\frac{k}{2}$ exploration phases succeeded, and enough matching
trials succeeded, then a $\gamma^{*}$-matching was found
within the last $\frac{k}{2}$ matching phases. This ensures that a $\gamma^{*}$-matching is
played during the $k$-th exploitation phase.
\begin{lem}[Exploitation Error Probability]
\label{lem:Convergence Lemma}Define the $k$-th exploitation error
event $E_{k}$ as the event where the actions $\boldsymbol{\tilde{a}}_{k^{*}}$
played in the $k$-th exploitation phase are not a $\gamma^{*}$-matching.
Let $k_0$ be large enough such that for all $k>k_0$
\begin{equation}
\varepsilon_{\lceil\frac{k}{2}\rceil}<\frac{\Delta}{4}\textnormal{ and   }1-\frac{\Bar{\tau}}{\lceil c_{2}\log (\frac{k}{2}+1)\rceil} - \frac{1 + \log k}{k/6}\ge \frac{3}{\sqrt{10}} .\label{eq:15}
\end{equation}
Then for all $k>k_0$ we have
\begin{equation}
\P\left(E_{k}\right)\leq 7NMe^{-\frac{c_{1}}{12}k}+e^{-\frac{3k}{10}}.\label{eq:16}
\end{equation}
\end{lem}
\begin{proof}
Define $E_{c,\ell}$ as the event where a matching existed in $G_{\ell}$ and was not found in the $\ell$-th matching phase. From Lemma \ref{Lem:MatchingTrial}
we know that if there is a matching in $G_{\ell}$, then the $\ell$-th
trial has success probability at least $1-\frac{\Bar{\tau}}{\lceil c_{2}\log (\ell+1)\rceil}$. 

Next we bound from below the number of trials we have between resets in order to find a $\gamma^{*}$ matching. We define $k_{w}\geq\left\lceil\frac{k}{2}\right\rceil$
as the first epoch since $\left\lceil\frac{k}{2}\right\rceil$ where a reset occurred (so
$\gamma_{k_{w}}=0$). In the worst case the algorithm resets in epoch $\left\lceil\frac{k}{2}\right\rceil-1$. Even still, the algorithm will reset again no later than
$k_{w}\leq \left\lceil\frac{k}{2}\right\rceil-1+ \left\lceil\frac{\left\lceil\frac{k}{2}\right\rceil-1}{3}\right\rceil \leq\frac{2}{3}k$.
The subsequent reset will then happen at $k_{w+1}$, where $k_{w+1}\leq\frac{2}{3}k+\left\lceil\frac{2}{9}k\right\rceil\leq\left\lceil\frac{8}{9}k\right\rceil<k$ for $k>9$.
We conclude that during the past $\frac{k}{2}$ epochs, there exists at least one full period (from reset to reset) with length at least
$\frac{k}{6}$.
Recall the definition of the $\ell$-th exploration error event $E_{e,\ell}$ in \eqref{eq:6}. Define the event $A_k=\bigcap_{\ell=\left\lceil \frac{k}{2}\right\rceil }^{k}\bar{E}_{e,\ell}$
for which $\bar{A}_k=\bigcup_{\ell=\left\lceil \frac{k}{2}\right\rceil }^{k}E_{e,\ell}$.
We define $\beta_{k_w}$ as the number of successful trials needed after reset $w$ to reach $\gamma_k \geq \gamma^{*}-\frac{\Delta}{4}$. Note that $\beta_{k_w} \leq 1 + \log k$ since no more than $1 + \log k$ steps of size $\varepsilon_{k_w} = \frac{1}{1+ \log k}$ are needed. 
Then for all $k>k_{0}$
\begin{align*}
\P(E_{k}\,|\,A_k)
&\underset{\left(a\right)}{\leq} \P\left(\gamma_{k}<\gamma^{*}-\frac{\Delta}{4}\,\bigg|\,A_k\right) \\
&\underset{\left(b\right)}{\leq}
\P\left(\sum_{\ell=k_{w}}^{k_{w+1}}\mathds{1}\left\{ \bar{E}_{c,\ell}\right\} <\beta_{k_w}\,\bigg|\,A_k\right) \\
&\underset{\left(c\right)}{\leq}
e^{-2\left(1-\frac{\Bar{\tau}}{\lceil c_{2}\log (\frac{k}{2}+1)\rceil}-\frac{1 +\log k}{k_{w+1} - k_w}\right)^{2}\left(k_{w+1}-k_{w}\right)}\\
&\underset{\left(d\right)}{\leq}e^{-\frac{3k}{10}} \numberthis \label{eq:17}
\end{align*}
where (a) follows since given $\bigcap_{\ell=\left\lceil \frac{k}{2}\right\rceil }^{k}\bar{E}_{e,\ell}$,
if $\gamma_{k}\geq\gamma^{*}-\frac{\Delta}{4}$ then a $\gamma^{*}$-matching
was found before the $k$-th exploitation phase and $E_{k}$ did not
occur. This follows since at the last success at $\ell\leq k$ we must have then had that for
all $n$
\begin{align*}
    \mu_{n,a_{n}}&\geq\mu_{n,a_{n}}^{\ell}-C_{n,a_{n}}^{\ell}\geq\gamma_{\ell}-2C_{n,a_{n}}^{\ell}\\
    &\geq \gamma^{*}-\frac{\Delta}{4}-\varepsilon_{k_w}-2C_{n,a_{n}}^{\ell}>\gamma^{*}-\Delta\numberthis \label{eq:18}
\end{align*}
which can only happen if $\mu_{n,a_{n}}\geq\gamma^{*}$. Inequality (b) in \eqref{eq:17} follows by noting that the probability that $\max_{\lceil \frac{k}{2}\rceil\le \ell \le k} \gamma_\ell < \gamma^{*}-\frac{\Delta}{4}$
with a constant step size (between resets) $\varepsilon_{k_w}$ implies fewer than $\left\lceil \frac{\gamma^{*}-\frac{\Delta}{4}}{\varepsilon_{k_w}}\right\rceil $
successful trials between $k_w$ and $k_{w+1}$.
Given $A_k$, in any trial $\ell \in [k_w,k_{w+1}]$ such that there have been fewer than $\left\lceil \frac{\gamma^{*}-\frac{\Delta}{4}}{\varepsilon_{k_w}}\right\rceil $ successes in $[k_w,\ell)$, at least one matching will exist in $G_\ell$ (an optimal matching $a_{n}^{*}$), since
\begin{equation}
\mu_{n,a_{n}^{*}}^{\ell}\geq\mu_{n,a_{n}^{*}}-C_{n,a_{n}^{*}}^{\ell}\geq\gamma^{*}-C_{n,a_{n}^{*}}^{\ell}\underset{(1)}{>}\gamma_{\ell}-C_{n,a_{n}^{*}}^{\ell}\label{eq:19}
\end{equation}
where (1) follows since for all $k>k_{0}$, $\varepsilon_{k_w}$ is sufficiently small such that $\gamma_{\ell}\leq\gamma^{*}-\frac{\Delta}{4}+\varepsilon_{k_w}<\gamma^{*}$.
Inequality (c) in \eqref{eq:17} follows from Lemma \ref{lem:baharavineq} with $p\triangleq1-\frac{\Bar{\tau}}{\lceil c_{2}\log (\frac{k}{2}+1)\rceil }$.
Inequality (d) follows from $k_{w+1}-k_{w}\geq\frac{k}{6}$ and \eqref{eq:15}.
Finally, \eqref{eq:16} is obtained by:
\begin{align*}
\P\left(E_{k}\right)&= \P\left(E_{k}\,|\bar{A}_{k}\right) \P\left(\bar{A}_{k}\right)+ \P\left(E_{k}\,|A_{k}\right) \P\left(A_{k}\right)\\
&\underset{\left(a\right)}\leq \bigg( 3NM\sum_{\ell=\left\lceil \frac{k}{2}\right\rceil }^{k}e^{-\frac{c_{1}}{6} \ell}\bigg)+e^{-\frac{3k}{10}}\\
&\underset{\left(b\right)}\leq 3NMe^{-\frac{c_{1}}{12}k} \left(\frac{1-e^{-\frac{c_{1}}{12}k}}{1-e^{-\frac{c_{1}}{6}}}\right)+e^{-\frac{3k}{10}}\\
&\leq 7NMe^{-\frac{c_{1}}{12}k}+e^{-\frac{3k}{10}}\numberthis \label{eq:21}
\end{align*}
where (a) is a union bound of $\bar{A}_k=\bigcup_{\ell=\left\lceil \frac{k}{2}\right\rceil }^{k}E_{e,\ell}$ using Lemma \ref{lem: exploration} together with \eqref{eq:17}, and (b) is a geometric sum. \qedhere
\end{proof}

\section{Numerical Simulations} \label{sec:Simulations}
We simulated two multi-armed bandit games with the following expected rewards matrices:
$$
 U_1=
\begin{bmatrix}
\frac{1}{2} & \frac{9}{10} & \frac{1}{10} & \frac{1}{4}\\[1pt]
\frac{1}{4} & \frac{1}{2} & \frac{1}{4} & \frac{1}{10}\\[1pt]
\frac{1}{10}  & \frac{1}{4} & \frac{1}{2} & \frac{1}{2} \\[1pt]
\frac{1}{10} & \frac{9}{10} & \frac{1}{4} & \frac{1}{2}
\end{bmatrix} 
$$
$$
U_2=\begin{bmatrix}
\frac{9}{10} & \frac{2}{5} & \frac{4}{5} & \frac{1}{10} & \frac{3}{10} & \frac{1}{20} & \frac{1}{5} & \frac{1}{10} & \frac{3}{10} & \frac{1}{5}\\[1pt]
\frac{4}{10} & \frac{3}{10} & \frac{3}{10} & \frac{1}{10} & \frac{1}{5} & \frac{3}{10} & \frac{2}{5} & \frac{2}{5} & \frac{3}{10} & \frac{2}{5}\\[1pt]
\frac{1}{10} & \frac{1}{20} & \frac{1}{10} & \frac{2}{5} & \frac{1}{10} & \frac{1}{5} & \frac{9}{10} & \frac{3}{10} & \frac{2}{5} & \frac{1}{10}\\[1pt]
\frac{1}{20} & \frac{1}{10} & \frac{9}{10} & \frac{1}{5} & \frac{9}{10} & \frac{3}{4} & \frac{1}{10} & \frac{9}{10} & \frac{1}{4} & \frac{1}{20}\\[1pt]
\frac{4}{5} & \frac{3}{10} & \frac{1}{10} & \frac{7}{10} & \frac{1}{10} & \frac{2}{5} & \frac{1}{20} & \frac{1}{5} & \frac{3}{4} & \frac{1}{20}\\[1pt]
\frac{2}{5} & \frac{1}{20} & \frac{3}{10} & \frac{7}{10} & \frac{1}{20} & \frac{1}{10} & \frac{1}{4} & \frac{3}{4} & \frac{3}{5} & \frac{1}{20}\\[1pt]
\frac{9}{10} & \frac{3}{10} & \frac{3}{10} & \frac{4}{5} & \frac{1}{10} & \frac{1}{4} & \frac{7}{10} & \frac{1}{20} & \frac{1}{5} & \frac{3}{10} \\[1pt]
\frac{3}{10} & \frac{1}{10} & \frac{2}{5}& \frac{1}{4} & \frac{1}{20} & \frac{9}{10} & \frac{1}{4} & \frac{1}{10} & \frac{1}{20} & \frac{2}{5} \\[1pt]
\frac{4}{5} & \frac{3}{4} & \frac{1}{10} & \frac{1}{5} & \frac{2}{5} & \frac{1}{20} & \frac{3}{10}  & \frac{1}{5} & \frac{1}{10} & \frac{1}{4} \\[1pt]
\frac{2}{5} & \frac{2}{5} & \frac{9}{10} & \frac{7}{10} & \frac{1}{4} & \frac{1}{5} & \frac{1}{20} & \frac{1}{10} & \frac{2}{5} & \frac{1}{4} \\[1pt]
\end{bmatrix}.
$$
Given expected rewards $\left\{ \mu_{n,i}\right\}$, the
rewards are generated as $r_{n,i}\left(t\right)=\mu_{n,i}+z_{n,i}\left(t\right)$
where $\left\{ z_{n,i}\left(t\right)\right\} $ are independent and
uniformly distributed on $\left[-0.05,0.05\right]$ for each $n,i$. The chosen
parameters were $c_{1}=1000$ and $c_{2}=2000$ and $c_{3}=4000$ for all experiments, and are chosen to ensure that the additive constant $C_0$ is small (since $k_0$ is small), as the exploration and matching phases are long enough from the beginning.

At the beginning of the $k$-th matching phase, each player played her action from the last exploitation phase if it is in  $\mathcal{E}_{n}^{k}$, or a random action from $\mathcal{E}_{n}^{k}$ otherwise. Although it has no effect on the theoretical bounds, it improved the performance in practice significantly. Another practical improvement was achieved by introducing a factor of 0.01 to the confidence intervals, which requires larger $c_1$ but does not affect the analysis otherwise. The step size sequence, that is updated only on resets, was chosen as $\varepsilon_k=\frac{0.2}{1+\log k}$.

In Fig. \ref{fig:totalRegretN4}, we present the total expected regret
versus time, averaged over 100 realizations, with $U_1$ as the expected reward matrix  ($N=4$). The shaded area denotes one standard deviation around the mean. This scenario has 24 matchings - 16
with minimal expected reward  $\frac{1}{10}$, 7 with $\frac{1}{4}$, and one optimal
matching with $\frac{1}{2}$.  It can be seen that
in all 100 experiments the players learned the max-min optimal matching by the end of the third epoch. This suggests that $k_{0}$ is much smaller than our theoretical bound. As expected, the regret scales (approximately) logarithmically. For comparison, the optimal sum of expected rewards for $U_1$ is 2.15, but the matching that achieves it has a minimal expected reward of $\frac{1}{4}$. Hence, a multi-player bandit algorithm that optimizes the expected sum of rewards will have regret $\Omega\left(\frac{T}{4}\right)$.
 
In Fig. \ref{fig:totalRegretN5}, we present the total expected regret
versus time, averaged over 100 realizations, with $U_2$ as the expected reward matrix ($N=10$). The shaded area denotes one standard deviation around the mean. In this scenario only 136 matchings out of the $10!=3628800$ are optimal with minimal utility of 0.4. There are 3798 matchings with 0.3, 16180 with 0.25, 62066 with 0.2,  785048 with 0.1 and 2761572 with 0.05. With more players and arms, $k_{0}$ is larger, but players still learn the max-min optimal matching by the sixth epoch. Again, the regret scales  (approximately) logarithmically as guaranteed by Theorem 1. For comparison, the optimal sum of expected rewards for $U_2$ is 7.35, but the matching that achieves it has a minimal expected reward of 0.3. Hence, a multi-player bandit algorithm that optimizes the expected sum of rewards will have regret $\Omega\left(\frac{T}{10}\right)$.
\begin{figure}[tbh]
\centering
\includegraphics[width=7cm,height=7cm,keepaspectratio]{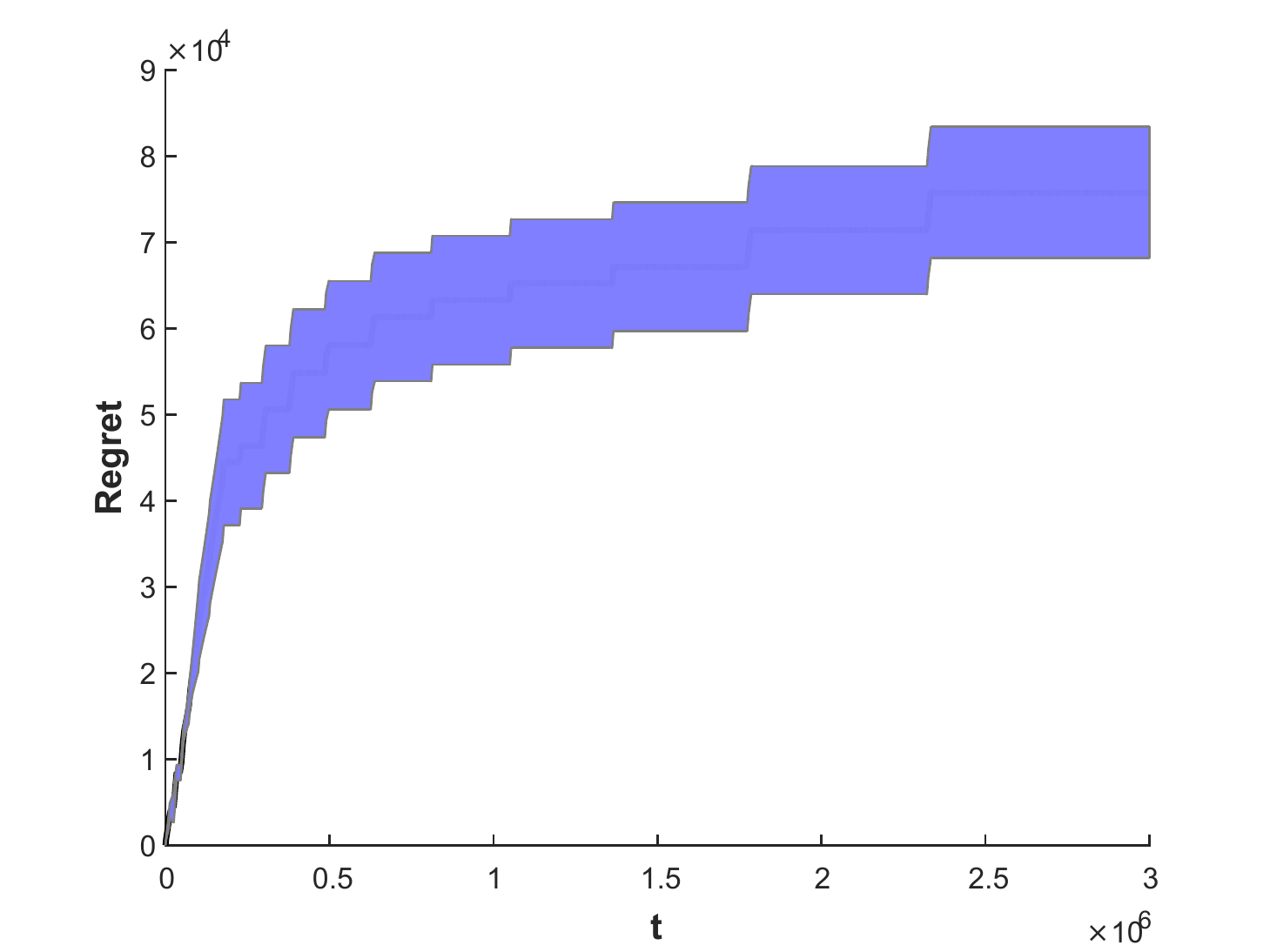}
\vspace{-.5cm}
\caption{Total regret as a function of time, averaged
over 100 experiments and with $U_1$ ($N=4$).\label{fig:totalRegretN4}}
\end{figure}
\begin{figure}[tbh]
\vspace{-.5cm}
\centering
\includegraphics[width=7cm,height=7cm,keepaspectratio]{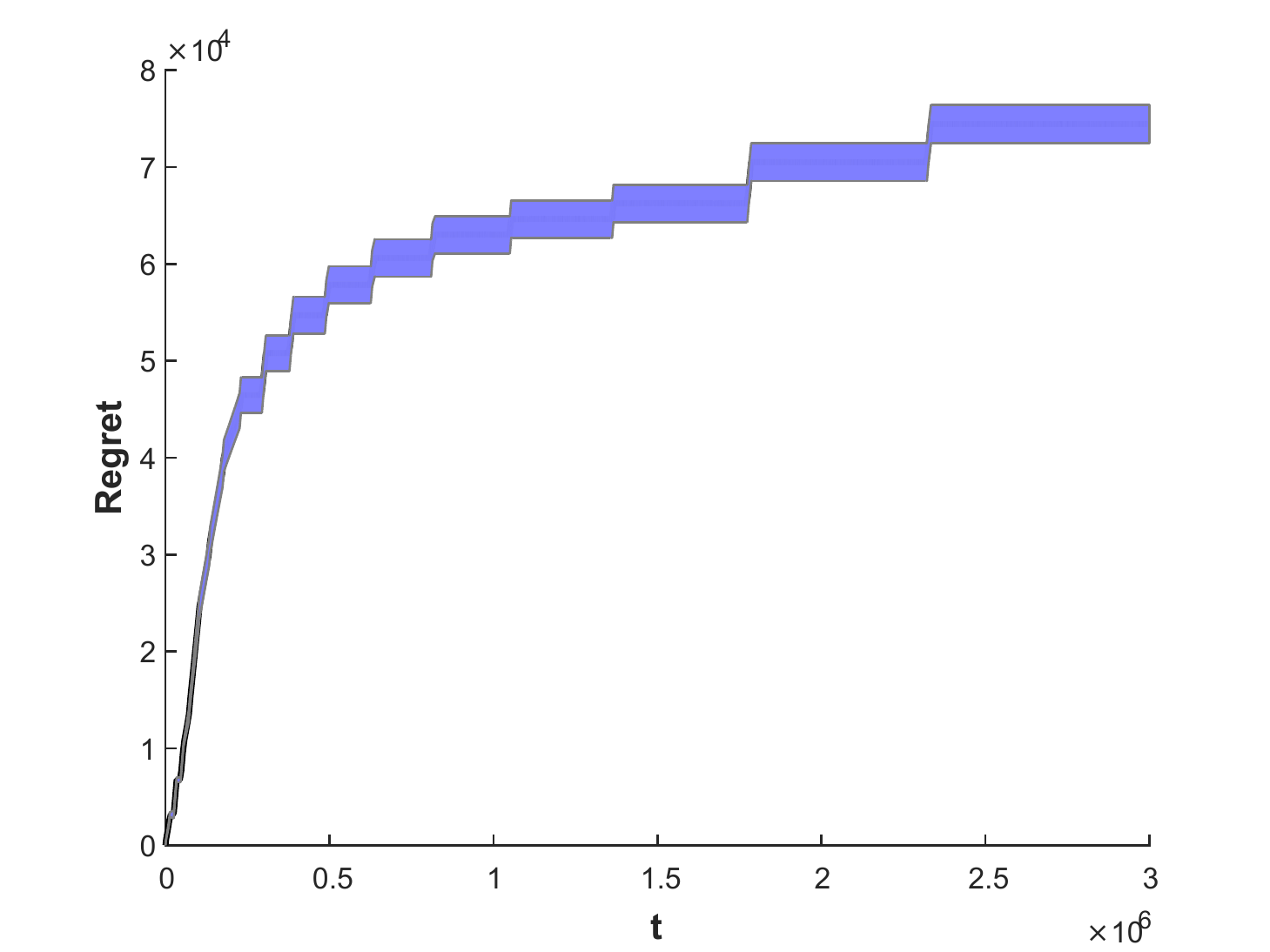}
\vspace{-.5cm}
\caption{Total  regret as a function of time, averaged
over 100 experiments and with $U_2$ ($N=10$).\label{fig:totalRegretN5}}
\vspace{-.5cm}
\end{figure}
\section{Conclusions and Future Work}
We studied a multi-player multi-armed bandit game where players cooperate to learn how to allocate arms, thought of as resources, so as to maximize the minimal expected reward received by any player. To allow for a meaningful notion of fairness, we employed the heterogeneous model where arms can have different expected rewards for each player.
Our algorithm operates in the restrictive setting of bandit feedback, where each player only observes the reward for the arm she played and cannot observe the actions or rewards of other players.
We proposed a novel fully distributed algorithm that achieves a near order optimal total expected regret of $O\left(\log \log T \log T\right)$, where $\log \log T$ can be improved to any increasing function of $T$.

It is still an open question whether a total expected regret of $O(\log T)$ is achievable in our scenario, when the problem parameters are unknown. Following our discussion on the additive constant $C_0$, an algorithm that achieves  $O(\log T)$ even with unknown parameters will expose the multiplicative factors of the $\log T$ in the regret bound, and their dependence on $\Delta$ and $M$. These factors are likely to be strongly affected by the time it takes the algorithm to find a matching. If some limited communication
is allowed, more sophisticated algorithms to distributedly compute
the matching are possible, based on gossip, message passing, or auctions \cite{bayati2008max, naparstek2016expected}. These algorithms do not need a consensus phase, eliminating the $M$ factor in \eqref{eq:4} and reducing the convergence time $\bar{\tau}$, but it is unclear if these approaches can achieve $O(\log T)$ regret.
Focusing on our setting, an interesting question is whether one can design a better distributed matching algorithm that can operate with no communication between players.

In some applications one is interested in guaranteeing a target QoS for each user that is ``good enough'' (see \cite{lai1984asymptotically,katz2020true} for the single-player case). This is a weaker requirement than max-min fairness between users. Hence, an interesting open question is whether better regret bounds for our multi-player bandit scenario can be obtained in this case. 

\section{Proof of Theorem 1}
Let $K$ be the number of epochs that start within the $T$ turns.
Since 
\begin{align}
T&\geq\sum_{k=1}^{K-1}\left(c_{1}\log k+c_{2}\log k+M+c_{3}\left(\frac{4}{3}\right)^{k}\right) \nonumber\\
&\geq3c_{3}\left(\left(\frac{4}{3}\right)^{K}-\frac{4}{3}\right)  
\end{align}
$K$ is upper bounded by $K\leq\log_{\frac{4}{3}}\left(\frac{T}{3c_{3}}+\frac{4}{3}\right)$.
Let $k_{0}$ be a constant epoch index that is large enough for
the bounds of Lemma \ref{lem: exploration}, Lemma \ref{lem:Convergence Lemma}, and inequality (c) in \eqref{eq:22} to hold. Intuitively,
this is the epoch after which the matching phase duration is long
enough, the step size  $\varepsilon_{k}$ is small enough, and the confidence intervals are sufficiently tight. Define $E_{k}$
as the event where a $\gamma^{*}$-matching is not played in the
$k$-th exploitation phase. We now bound the total expected regret
of epoch $k>k_{0}$, denoted by ${R}_{k}$:
\begin{align*}
    {R}_{k}&\leq M + \left(c_{1}+c_{2}\right)\log (k+1)+\P\left(E_{k}\right)c_{3}\left(\frac{4}{3}\right)^{k} +3\nonumber\\
    & \underset{\left(a\right)}{\leq} M + (c_{1}+c_{2})\log (k+1)+3\nonumber\\
    &\hspace{1cm} +\left(7NMe^{-\frac{c_{1}}{12}k}+e^{-\frac{3k}{10}}\right)c_{3}\left(\frac{4}{3}\right)^{k}\nonumber\\
    &\underset{\left(b\right)}{\leq} M+3+(c_{1}+c_{2})\log (k+1)+8NMc_{3}\beta^{k}\nonumber\\
    &\underset{\left(c\right)}{\leq}M+2(c_{1}+c_{2})\log k \numberthis \label{eq:22}
\end{align*}
where (a) uses Lemma \ref{lem:Convergence Lemma}, (b) follows for
some constant $\beta<1$ since $e^{-\frac{3}{10}}<\frac{3}{4}$ and
$c_{1}\geq4$ and (c) follows for $k>k_{0}$. We conclude that, for some additive constant $C_0$,
\begin{align*}
R(T)&=\sum_{k=1}^{K}{R}_{k} \underset{\left(a\right)}{\leq}MK+2\sum_{k=k_{0}+1}^{K}\left(c_{1}+c_{2}\right)\log k\\&\hspace{.5cm}+\sum_{k=1}^{k_{0}}\left(\left(c_{1}+c_{2}\right)\log (k+1)+c_{3}\left(\frac{4}{3}\right)^{k}+3\right)
\\
&{\leq} C_0+MK+ 2\left(c_{1}+c_{2}\right)K\log K \numberthis\label{eq:23}
\end{align*}
where (a) follows by completing the last epoch to a full epoch
which only increases $R(T)$, and by using \eqref{eq:22}. Then, we obtain \eqref{eq:4} by upper bounding $K\leq\log_{\frac{4}{3}}\left(\frac{T}{3c_{3}}+\frac{4}{3}\right)\leq\log_{\frac{4}{3}}\left(\frac{T}{c_{3}}\right)$, where the second inequality is only used to simplify the expression, and holds for all $T\geq2c_3$.
\vspace{-.25cm}
\section*{Acknowledgements}
\vspace{-.25cm}
The authors gratefully acknowledge funding from the Koret Foundation grant for Smart Cities and Digital Living, NSF GRFP, Alcatel-Lucent Stanford Graduate Fellowship, and ISF grant 1644/18.

\bibliographystyle{icml2020.bst}
\bibliography{FairBandit}

\begin{thebibliography}{40}
\providecommand{\natexlab}[1]{#1}
\providecommand{\url}[1]{\texttt{#1}}
\expandafter\ifx\csname urlstyle\endcsname\relax
  \providecommand{\doi}[1]{doi: #1}\else
  \providecommand{\doi}{doi: \begingroup \urlstyle{rm}\Url}\fi

\bibitem[Alatur et~al.(2019)Alatur, Levy, and Krause]{alatur2019multi}
Alatur, P., Levy, K.~Y., and Krause, A.
\newblock Multi-player bandits: The adversarial case.
\newblock \emph{arXiv preprint arXiv:1902.08036}, 2019.

\bibitem[Anandkumar et~al.(2011)Anandkumar, Michael, Tang, and
  Swami]{Anandkumar2011}
Anandkumar, A., Michael, N., Tang, A.~K., and Swami, A.
\newblock Distributed algorithms for learning and cognitive medium access with
  logarithmic regret.
\newblock \emph{IEEE Journal on Selected Areas in Communications}, 29\penalty0
  (4):\penalty0 731--745, 2011.

\bibitem[Asadpour \& Saberi(2010)Asadpour and
  Saberi]{asadpour2010approximation}
Asadpour, A. and Saberi, A.
\newblock An approximation algorithm for max-min fair allocation of indivisible
  goods.
\newblock \emph{SIAM Journal on Computing}, 39\penalty0 (7):\penalty0
  2970--2989, 2010.

\bibitem[Avner \& Mannor(2014)Avner and Mannor]{Avner2014}
Avner, O. and Mannor, S.
\newblock Concurrent bandits and cognitive radio networks.
\newblock In \emph{Joint European Conference on Machine Learning and Knowledge
  Discovery in Databases}, pp.\  66--81, 2014.

\bibitem[Avner \& Mannor(2016)Avner and Mannor]{Avner2016}
Avner, O. and Mannor, S.
\newblock Multi-user lax communications: a multi-armed bandit approach.
\newblock In \emph{INFOCOM 2016-The 35th Annual IEEE International Conference
  on Computer Communications, IEEE}, pp.\  1--9, 2016.

\bibitem[Bar-On \& Mansour(2019)Bar-On and Mansour]{bar2019individual}
Bar-On, Y. and Mansour, Y.
\newblock Individual regret in cooperative nonstochastic multi-armed bandits.
\newblock In \emph{Advances in Neural Information Processing Systems}, pp.\
  3110--3120, 2019.

\bibitem[Bayati et~al.(2008)Bayati, Shah, and Sharma]{bayati2008max}
Bayati, M., Shah, D., and Sharma, M.
\newblock Max-product for maximum weight matching: Convergence, correctness,
  and lp duality.
\newblock \emph{IEEE Transactions on Information Theory}, 54\penalty0
  (3):\penalty0 1241--1251, 2008.

\bibitem[Besson \& Kaufmann(2018)Besson and Kaufmann]{Besson2018}
Besson, L. and Kaufmann, E.
\newblock Multi-player bandits revisited.
\newblock In \emph{Algorithmic Learning Theory}, pp.\  56--92, 2018.

\bibitem[Bistritz \& Leshem(2018)Bistritz and Leshem]{bistritz2018distributed}
Bistritz, I. and Leshem, A.
\newblock Distributed multi-player bandits-a game of thrones approach.
\newblock In \emph{Advances in Neural Information Processing Systems}, pp.\
  7222--7232, 2018.

\bibitem[Boursier \& Perchet(2019)Boursier and Perchet]{boursier2019sic}
Boursier, E. and Perchet, V.
\newblock {SIC-MMAB}: synchronisation involves communication in multiplayer
  multi-armed bandits.
\newblock In \emph{Advances in Neural Information Processing Systems}, pp.\
  12048--12057, 2019.

\bibitem[Boursier et~al.(2019)Boursier, Perchet, Kaufmann, and
  Mehrabian]{boursier2019practical}
Boursier, E., Perchet, V., Kaufmann, E., and Mehrabian, A.
\newblock A practical algorithm for multiplayer bandits when arm means vary
  among players.
\newblock \emph{arXiv preprint arXiv:1902.01239}, 2019.

\bibitem[Bubeck et~al.(2012)Bubeck, Cesa-Bianchi, et~al.]{bubeck2012regret}
Bubeck, S., Cesa-Bianchi, N., et~al.
\newblock Regret analysis of stochastic and nonstochastic multi-armed bandit
  problems.
\newblock \emph{Foundations and Trends{\textregistered} in Machine Learning},
  5\penalty0 (1):\penalty0 1--122, 2012.

\bibitem[Bubeck et~al.(2019)Bubeck, Li, Peres, and Sellke]{bubeck2019non}
Bubeck, S., Li, Y., Peres, Y., and Sellke, M.
\newblock Non-stochastic multi-player multi-armed bandits: Optimal rate with
  collision information, sublinear without.
\newblock \emph{arXiv preprint arXiv:1904.12233}, 2019.

\bibitem[Cohen et~al.(2017)Cohen, H{\'e}liou, and Mertikopoulos]{Cohen2017}
Cohen, J., H{\'e}liou, A., and Mertikopoulos, P.
\newblock Learning with bandit feedback in potential games.
\newblock In \emph{Proceedings of the 31th International Conference on Neural
  Information Processing Systems}, 2017.

\bibitem[Darak \& Hanawal(2019)Darak and Hanawal]{darak2019multi}
Darak, S.~J. and Hanawal, M.~K.
\newblock Multi-player multi-armed bandits for stable allocation in
  heterogeneous ad-hoc networks.
\newblock \emph{IEEE Journal on Selected Areas in Communications}, 37\penalty0
  (10):\penalty0 2350--2363, 2019.

\bibitem[Evirgen \& Kose(2017)Evirgen and Kose]{Evirgen2017}
Evirgen, N. and Kose, A.
\newblock The effect of communication on noncooperative multiplayer multi-armed
  bandit problems.
\newblock In \emph{arXiv preprint arXiv:1711.01628, 2017}, 2017.

\bibitem[Hanawal \& Darak(2018)Hanawal and Darak]{hanawal2018multi}
Hanawal, M.~K. and Darak, S.~J.
\newblock Multi-player bandits: A trekking approach.
\newblock \emph{arXiv preprint arXiv:1809.06040}, 2018.

\bibitem[Hoeffding(1994)]{hoeffding1994probability}
Hoeffding, W.
\newblock Probability inequalities for sums of bounded random variables.
\newblock In \emph{The Collected Works of Wassily Hoeffding}, pp.\  409--426.
  Springer, 1994.

\bibitem[Jabbari et~al.(2017)Jabbari, Joseph, Kearns, Morgenstern, and
  Roth]{jabbari2017fairness}
Jabbari, S., Joseph, M., Kearns, M., Morgenstern, J., and Roth, A.
\newblock Fairness in reinforcement learning.
\newblock In \emph{Proceedings of the 34th International Conference on Machine
  Learning-Volume 70}, pp.\  1617--1626. JMLR. org, 2017.

\bibitem[Joseph et~al.(2016)Joseph, Kearns, Morgenstern, and
  Roth]{joseph2016fairness}
Joseph, M., Kearns, M., Morgenstern, J.~H., and Roth, A.
\newblock Fairness in learning: Classic and contextual bandits.
\newblock In \emph{Advances in Neural Information Processing Systems}, pp.\
  325--333, 2016.

\bibitem[Kalathil et~al.(2014)Kalathil, Nayyar, and Jain]{Kalathil2014}
Kalathil, D., Nayyar, N., and Jain, R.
\newblock Decentralized learning for multiplayer multiarmed bandits.
\newblock \emph{IEEE Transactions on Information Theory}, 60\penalty0
  (4):\penalty0 2331--2345, 2014.

\bibitem[Katz-Samuels \& Jamieson(2020)Katz-Samuels and Jamieson]{katz2020true}
Katz-Samuels, J. and Jamieson, K.
\newblock The true sample complexity of identifying good arms.
\newblock In \emph{International Conference on Artificial Intelligence and
  Statistics}, pp.\  1781--1791, 2020.

\bibitem[Lai et~al.(2008)Lai, Jiang, and Poor]{Lai2008}
Lai, L., Jiang, H., and Poor, H.~V.
\newblock Medium access in cognitive radio networks: A competitive multi-armed
  bandit framework.
\newblock In \emph{Signals, Systems and Computers, 2008 42nd Asilomar
  Conference on}, pp.\  98--102, 2008.

\bibitem[Lai \& Robbins(1984)Lai and Robbins]{lai1984asymptotically}
Lai, T.~L. and Robbins, H.
\newblock Asymptotically optimal allocation of treatments in sequential
  experiments.
\newblock \emph{Design of Experiments: Ranking and Selection}, pp.\  127--142,
  1984.

\bibitem[Lai \& Robbins(1985)Lai and Robbins]{Lai1985}
Lai, T.~L. and Robbins, H.
\newblock Asymptotically efficient adaptive allocation rules.
\newblock \emph{Advances in applied mathematics}, 6\penalty0 (1):\penalty0
  4--22, 1985.

\bibitem[Liu et~al.(2013)Liu, Liu, and Zhao]{Liu2013}
Liu, H., Liu, K., and Zhao, Q.
\newblock Learning in a changing world: Restless multiarmed bandit with unknown
  dynamics.
\newblock \emph{IEEE Transactions on Information Theory}, 59\penalty0
  (3):\penalty0 1902--1916, 2013.

\bibitem[Liu \& Zhao(2010)Liu and Zhao]{Liu2010}
Liu, K. and Zhao, Q.
\newblock Distributed learning in multi-armed bandit with multiple players.
\newblock \emph{IEEE Transactions on Signal Processing}, 58\penalty0
  (11):\penalty0 5667--5681, 2010.

\bibitem[Liu et~al.(2019)Liu, Mania, and Jordan]{liu2019competing}
Liu, L.~T., Mania, H., and Jordan, M.~I.
\newblock Competing bandits in matching markets.
\newblock \emph{arXiv preprint arXiv:1906.05363}, 2019.

\bibitem[Magesh \& Veeravalli(2019)Magesh and Veeravalli]{magesh2019multi}
Magesh, A. and Veeravalli, V.~V.
\newblock Multi-player multi-armed bandits with non-zero rewards on collisions
  for uncoordinated spectrum access.
\newblock \emph{arXiv preprint arXiv:1910.09089}, 2019.

\bibitem[Mo \& Walrand(2000)Mo and Walrand]{mo2000fair_alpha}
Mo, J. and Walrand, J.
\newblock Fair end-to-end window-based congestion control.
\newblock \emph{IEEE/ACM Transactions on networking}, \penalty0 (5):\penalty0
  556--567, 2000.

\bibitem[Naparstek \& Leshem(2016)Naparstek and Leshem]{naparstek2016expected}
Naparstek, O. and Leshem, A.
\newblock Expected time complexity of the auction algorithm and the push
  relabel algorithm for maximum bipartite matching on random graphs.
\newblock \emph{Random Structures \& Algorithms}, 48\penalty0 (2):\penalty0
  384--395, 2016.

\bibitem[Nayyar et~al.(2016)Nayyar, Kalathil, and Jain]{Nayyar2016}
Nayyar, N., Kalathil, D., and Jain, R.
\newblock On regret-optimal learning in decentralized multi-player multi-armed
  bandits.
\newblock \emph{IEEE Transactions on Control of Network Systems}, PP\penalty0
  (99):\penalty0 1--1, 2016.

\bibitem[Radunovic \& Le~Boudec(2007)Radunovic and
  Le~Boudec]{radunovic2007unified}
Radunovic, B. and Le~Boudec, J.-Y.
\newblock A unified framework for max-min and min-max fairness with
  applications.
\newblock \emph{IEEE/ACM Transactions on networking}, 15\penalty0 (5):\penalty0
  1073--1083, 2007.

\bibitem[Rosenski et~al.(2016)Rosenski, Shamir, and Szlak]{Rosenski2016}
Rosenski, J., Shamir, O., and Szlak, L.
\newblock Multi-player bandits--a musical chairs approach.
\newblock In \emph{International Conference on Machine Learning}, pp.\
  155--163, 2016.

\bibitem[Sankararaman et~al.(2019)Sankararaman, Ganesh, and
  Shakkottai]{sankararaman2019social}
Sankararaman, A., Ganesh, A., and Shakkottai, S.
\newblock Social learning in multi agent multi armed bandits.
\newblock \emph{Proceedings of the ACM on Measurement and Analysis of Computing
  Systems}, 3\penalty0 (3):\penalty0 1--35, 2019.

\bibitem[Tibrewal et~al.(2019)Tibrewal, Patchala, Hanawal, and
  Darak]{tibrewal2019distributed}
Tibrewal, H., Patchala, S., Hanawal, M.~K., and Darak, S.~J.
\newblock Distributed learning and optimal assignment in multiplayer
  heterogeneous networks.
\newblock In \emph{IEEE INFOCOM 2019-IEEE Conference on Computer
  Communications}, pp.\  1693--1701. IEEE, 2019.

\bibitem[Vakili et~al.(2013)Vakili, Liu, and Zhao]{Vakili2013}
Vakili, S., Liu, K., and Zhao, Q.
\newblock Deterministic sequencing of exploration and exploitation for
  multi-armed bandit problems.
\newblock \emph{IEEE Journal of Selected Topics in Signal Processing},
  7\penalty0 (5):\penalty0 759--767, 2013.

\bibitem[Wei et~al.(2015)Wei, Iyer, Wang, Bai, and Bilmes]{wei2015mixed}
Wei, K., Iyer, R.~K., Wang, S., Bai, W., and Bilmes, J.~A.
\newblock Mixed robust/average submodular partitioning: Fast algorithms,
  guarantees, and applications.
\newblock In \emph{Advances in Neural Information Processing Systems}, pp.\
  2233--2241, 2015.

\bibitem[Zehavi et~al.(2013)Zehavi, Leshem, Levanda, and
  Han]{zehavi2013weighted}
Zehavi, E., Leshem, A., Levanda, R., and Han, Z.
\newblock Weighted max-min resource allocation for frequency selective
  channels.
\newblock \emph{IEEE transactions on signal processing}, 61\penalty0
  (15):\penalty0 3723--3732, 2013.

\bibitem[Zemel et~al.(2013)Zemel, Wu, Swersky, Pitassi, and
  Dwork]{zemel2013learning}
Zemel, R., Wu, Y., Swersky, K., Pitassi, T., and Dwork, C.
\newblock Learning fair representations.
\newblock In \emph{International Conference on Machine Learning}, pp.\
  325--333, 2013.

\end{thebibliography}

\end{document}